\documentclass[10pt]{article}
\usepackage{spconf}

\usepackage{blindtext}
\usepackage{comment}

\usepackage{amsmath}
\usepackage{amssymb}
\usepackage{amsfonts}
\usepackage{amsthm}
\usepackage{algorithm}
\usepackage{algorithmic}
\usepackage{nicefrac}
\usepackage{bm}
\usepackage{hyperref}
\usepackage{cite}
\usepackage[utf8]{inputenc} 
\usepackage[T1]{fontenc}
\usepackage{graphicx}
\usepackage[font=small]{caption}
\usepackage{tikz}
\usetikzlibrary{positioning}
\usepackage{pgfplots}
\pgfplotsset{compat=1.3}

\input{mysymbol.sty}

\newtheorem{Problem}{Problem}
\newtheorem{Theorem}{Theorem}

\newcommand{\letterref}[2]{\hyperref[#1]{\ref*{#1}#2}}

\def\ER{Erd\H{o}s-R\'{e}nyi }
\def\BA{Barab\'{a}si-Albert }

\title{Blind inference of centrality rankings from graph signals}
\name{T. Mitchell Roddenberry, Santiago Segarra\thanks{Emails: \href{mailto:mitch@rice.edu}{mitch@rice.edu}, \href{mailto:segarra@rice.edu}{segarra@rice.edu}.}}
\address{Department of Electrical and Computer Engineering, Rice University}

\begin{document}

\ninept


\maketitle

\begin{abstract}
  We study the blind centrality ranking problem, where our goal is to infer the eigenvector centrality ranking of nodes solely from nodal observations, i.e., without information about the topology of the network.
  We formalize these nodal observations as graph signals and model them as the outputs of a network process on the underlying (unobserved) network.
  A simple spectral algorithm is proposed to estimate the leading eigenvector of the associated adjacency matrix, thus serving as a proxy for the centrality ranking.
  A finite rate performance analysis of the algorithm is provided, where we find a lower bound on the number of graph signals needed to correctly rank (with high probability) two nodes of interest.
  We then specialize our general analysis for the particular case of dense \ER graphs, where existing graph-theoretical results can be leveraged.
  Finally, we illustrate the proposed algorithm via numerical experiments in synthetic and real-world networks, making special emphasis on how the network features influence the performance.
\end{abstract}


\section{Introduction}\label{sec:intro}

As the prevalence of complex, structured data has exploded in recent years, so has the analysis of such data using graph-based representations \cite{Strogatz2001,Newman2010,Jackson2010}.
Indeed, abstracting relational structures as graphs has become an ever-increasing paradigm in science and engineering.
In this context, network science aims to understand such structures, often via analysis of their algebraic properties when represented as matrices.

In any network, the topology determines an influence structure among the nodes or agents. 
Identifying the most important nodes in a network helps in explaining the network’s dynamics, e.g., migration in biological networks~\cite{garroway_2008_applications}, as well as in designing optimal ways to externally influence the network, e.g., vulnerability to attacks~\cite{holme_2002_attack}.
Node centrality measures are tools designed to identify such important agents. However, node importance is a rather vague concept and can be interpreted in various ways, giving rise to multiple coexisting centrality measures, some of the most common being
closeness~\cite{beauchamp_1965_improved}, betweenness~\cite{freeman_1977_set, segarra_2016_stability} and eigenvector~\cite{bonacich_1972_factoring} centrality.
In this latter measure -- of special interest to this paper -- the importance of a node is computed as a function of the importance of its neighbors.

Computation of the eigenvector centrality requires complete knowledge of the graph being studied.
However, graphs are often difficult or infeasible to by fully observed, especially in large-scale settings.
In these situations, we might rely on data supported on the nodes (that we denominate \emph{graph signals}) to infer network properties.
In this paper, we seek to answer the question: \textit{Under what conditions can one rank the nodes of a graph according to their eigenvector centrality, without observing the edges of the graph but given only a set of graph signals?}

As a motivating example, consider the observation of opinions of individuals (graph signals) in a social network. Intuitively, even without having access to the social connections of this group of people, we might be able to determine the most central actors by tracing back the predominating opinions.
In this paper, we analyze in which cases this intuition holds.

\vspace{1mm}
\noindent 
{\bf Related work.} The standard paradigm for network inference from nodal observations aims to infer the complete graph structure, also known as \textit{network topology inference}.
Network topology inference has been studied from a statistical perspective where each node is a random variable, and edges reflect the covariance structure of the ensemble of random variables
\cite{friedman_2008_sparse,lake_2010_discovering,meinshausen_2006_high,egilmez_2017_graph,shen_2017_kernel}.
Additionally, graph signal processing methods have arisen recently, which infer network topology by assuming the observed signals are the output of some underlying network process \cite{dong2019tutorial,kalofolias_smooth_2016,segarra_topo_2017,mateos_2019_connecting,zhu2019consensus}.
Our work differs from this line of research, in that we do not infer the network structure, but rather the centrality ranking of the nodes.
This latter feature, being a coarse description of the network, can be inferred with less samples than those needed to recover the detailed graph structure.

In this same spirit, recent work has considered the inference of coarse network descriptors from graph signals, but have exclusively focused on community detection.
More precisely, \cite{schaub2018blind,wai2018blind,wai2019hidden} provide algorithms and statistical guarantees for community detection on a single graph from sampled graph signals, whereas \cite{schaub2019partition,Hoffmann2018} analyze this problem for ensembles of graphs drawn from a latent random graph model.
Lastly, in terms of eigenvector centrality estimation from partial data, \cite{ruggeri2019centrality} considers the case of missing edges in the graph of interest, but does not rely on node data as we propose here.
In the direction of centrality inference from data, \cite{shao2017centrality} is the most similar to our work.
However, the authors are concerned with temporal data driven by consensus dynamics, which they use to infer an ad hoc temporal centrality. Our approach does not depend on temporal structure, and infers the classical eigenvector centrality instead.

\vspace{1mm}
\noindent 
{\bf Contributions.}
The contributions of this paper are threefold. First, we provide a simple spectral method to estimate the eigenvector centrality ranking of the nodes. Second, and most importantly, we provide theoretical guarantees for the proposed method. In particular, we determine the number of samples needed for a desired resolution in the centrality ranking.
Finally, we particularize our general theoretical analysis to the case of \ER graphs and showcase our findings via illustrative numerical experiments.

\section{Notation and background}

\noindent 
{\bf Graphs and eigenvector centrality.}
An undirected \emph{graph} ${\cal G}$ consists of a set ${\ccalV}$ of $n := |${\ccalV}$|$ nodes, and a set ${\ccalE}\subseteq{\ccalV}\times{\ccalV}$ of edges, corresponding to unordered pairs of elements in ${\cal V}$. 
We commonly index the nodes with the integers $1,2,\ldots,n$.
We then encode the graph structure with the (symmetric) adjacency matrix ${\bbA}\in{\mathbb R}^{n\times n}$, such that $A_{ij} = A_{ji} = 1$ for all $(i,j) \in \ccalE$, and $A_{ij}=0$ otherwise.

For a graph with adjacency matrix ${\bbA}$, we consider the eigenvector centrality~\cite{bonacich_1972_factoring} given by the leading eigenvector of ${\bbA}$.
That is, if ${\bbA}$ has eigenvalue decomposition ${\bbA}=\sum_{i=1}^n\lambda_i{\bbv}_i{\bbv}_i^\top$, where $\lambda_1\geq\lambda_2\geq\cdots\geq\lambda_n$, the eigenvector centrality of node $j$ is given by the $j^{\rm th}$ entry of ${\bbv}_1$.
For notational convenience, we denote this leading eigenvector by $\bbu$, so that the centrality value of node $j$ is given by $u_j$.

Often, we are not concerned with the precise centrality value of each node, but rather a rank ordering based on these values.
In this regard, we define the \textit{centrality rank} as
\begin{equation}
\label{eq:rank}
r_i=\big|\{j\in{\cal N}\colon u_j\geq u_i\}\big|,
\end{equation}
so the most central node has $r=1$ (in the absence of a tie) and the least central node has $r=n$.
Note that if two nodes have identical centrality values they will have the same centrality rank.
For any two nodes $i,j\in{\cal N}$, two rankings $r$ and $r'$ preserve the relative order of $i,j$ when $r_i\geq r_j \text{ if and only if } r'_i\geq r'_j$.
For convenience, we say that two vectors ${\bbu}$ and $\widehat{\bbu}$ preserve the relative order of two nodes, where it is implicit that this refers to the ranking induced by each vector.

\vspace{1mm}
\noindent
{\bf Graphs signals and graph filters.}
\textit{Graph signals} are defined as real-valued functions on the node set, i.e., $x\colon{\ccalV}\to{\reals}$.
Given an indexing of the nodes, graph signals can be represented as vectors in ${\reals}^n$, such that $x_i=x(i)$.
A \emph{graph filter} $\ccalH(\bbA)$ of order $T$ is a linear map between graph signals that can be expressed as a matrix polynomial in $\bbA$ of degree $T$,
\begin{equation}\label{eq:graphfilter}
{\ccalH}\left({\bbA}\right)=\sum_{k=0}^T\gamma_k{\bbA}^k:=\sum_{i=0}^T{\ccalH}(\lambda_i){\bbv}_i{\bbv}_i^\top,
\end{equation}
where ${\ccalH}(\lambda)$ is the extension of the matrix polynomial ${\cal H}$ to scalars.
Properly normalized and combined with a set of appropriately chosen filter coefficients $\gamma_k$, the transformation ${\cal H}\left({\bbA}\right)$ can account for a range of interesting dynamics including consensus~\cite{Olfati-Saber2007}, random walks, and diffusion~\cite{Masuda2017}.

\section{Blind centrality inference}\label{sec:bci}

Our goal is to infer the centrality rankings of the nodes without ever observing the graph structure, but rather exclusively relying on the observation of graph signals. 
Naturally, these graph signals must be shaped by the underlying unknown graph for the described inference task to be feasible.
In particular, we model the observed graph signals $\{\bby_i\}_{i=1}^N$ as being the output of graph filters excited by (unobserved) white noise, i.e., 
\begin{equation}\label{eq:graphfilter_2}
\bby_i = {\ccalH}\left({\bbA}\right) \bbw_i = \sum_{k=0}^T\gamma_k{\bbA}^k \bbw_i,
\end{equation}
where $\mathbb{E}[\bbw]={\bm 0}$, $\mathbb{E}[\bbw \bbw^\top]={\bbI}$, and $\gamma_k \geq 0$ for all $k$.
The rationale behind the inputs $\bbw_i$ being white is that we are interested in a situation where the original signal (e.g., the initial opinion of each individual in a social network) is uncorrelated among agents. In this way, the correlation structure of the output is exclusively imposed by exchanges among neighboring agents as opposed to being driven by features that are exogenous to the graph. Regarding the form of the filter ${\ccalH}\left({\bbA}\right)$, we only require non-negative coefficients $\gamma_k \geq 0$. 
Since successive powers of $\bbA$ aggregate information in neighborhoods of increasing sizes, $\gamma_k \geq 0$ simply imposes the notion that agents are positively influenced by the signals in their neighborhoods, e.g., in a social network the opinions of individuals grow similar to those in their surroundings.

With this notation in place, we formally state our problem:

\begin{Problem}\label{Prob:main}
Given the observation of $N$ signals $\{\bby_i\}_{i=1}^N$ generated as in~\eqref{eq:graphfilter_2}, estimate the node centrality ranking induced by $\bbu$, the leading eigenvector of the (unobserved) adjacency matrix $\bbA$.
\end{Problem}

Our proposed method to solve Problem~\ref{Prob:main} is summarized in Algorithm~\ref{alg:bci}, where we simply compute the sample covariance $\widehat{\bbC}_y^N$ of the observed outputs and use the ranking induced by its leading eigenvector as a proxy for the true centrality ranking. 
To see why this simple spectral method is reasonable, notice that the (population) covariance of the i.i.d. outputs $\bby_i$ is given by 
\begin{equation}\label{eq:covariance}
\bbC_y = \mathbb{E}[{\ccalH}\left({\bbA}\right) \bbw_i \bbw^\top_i {\ccalH}\left({\bbA}\right)^\top] = {\cal H}({\bbA})^2,
\end{equation}
where we used the facts that $\bbw_i$ is white and $\bbA$ is symmetric. 
Thus, from \eqref{eq:covariance} it follows that $\bbA$ and $\bbC_y$ share the same set of eigenvectors.
Moreover, since $\gamma_k \geq 0$ in \eqref{eq:graphfilter_2}, it must be the case that $\bbu$ is also the \emph{leading} eigenvector of $\bbC_y$. 
In this way, as $N$ grows and the sample covariance in \eqref{eq:alg1-cov} approaches $\bbC_y$, it is evident that $\widehat{\bbu}$ as computed in Algorithm~\ref{alg:bci} is a consistent estimator of the true centrality eigenvector $\bbu$.
Hence, for a large enough sample size $N$, Algorithm~\ref{alg:bci} is guaranteed to return the true centrality ranking.
However, the true practical value lies in the finite rate analysis, i.e., given a \emph{finite} sample size $N$ and two nodes of interest, when can we confidently state that we have recovered the true centrality ordering between these nodes?
We answer this question next.

\begin{algorithm}[tb]
  \caption{Blind centrality inference}\label{alg:bci}
  \begin{algorithmic}[1]
    \STATE \textbf{INPUT:} graph signals $\{{\bby}_i\}_{i=1}^N$
    \STATE Compute the sample covariance matrix
    \begin{equation}
      \label{eq:alg1-cov}
      \widehat{\bbC}_y^N:=\frac{1}{N}\sum_{i=1}^N{\bby}_i{\bby}_i^\top.
    \end{equation}
    \STATE Compute the leading eigenvector $\widehat{\bbu}$ of $\widehat{\bbC}_y^N$
    \STATE \textbf{OUTPUT:} Centrality ranking induced by $\widehat{\bbu}$
  \end{algorithmic}
\end{algorithm}

\begin{figure*}[t]
  \centering
  \resizebox{\linewidth}{!}{\includegraphics[width=\linewidth]{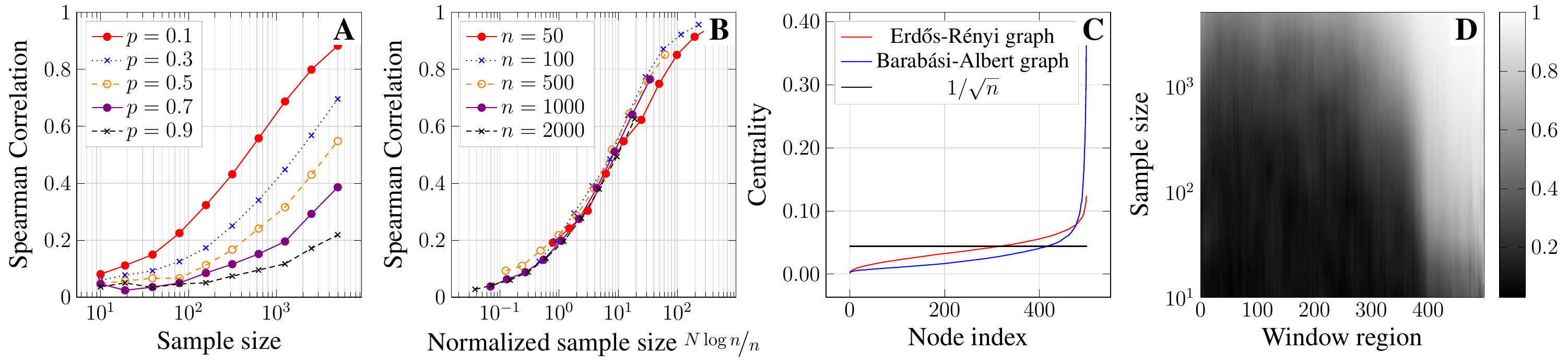}}
\vspace{-4mm}
  \caption{Blind centrality inference on synthetic graphs. Each plot is the average over 10 runs.
    (A) Spearman correlation for fixed graph size $n=500$ with varying edge probability $p$.
    (B) Spearman correlation for edge probability $p=4\log n/n$ with varying graph size $n$, plotted against the normalized sample size $N\log n/n$.
    (C) Comparison of eigenvector centrality for \ER graph with $p=\log n/n$ and \BA graph with $m=m_0=3$, where both graphs are of size $n=500$. Notice that both sets of parameters yield graphs of similar density.
    (D) Blind centrality inference for \BA graphs with $m=m_0=4$ and $n=500$.
    For each sample size, the Spearman correlation is evaluated for a windowed region ($\text{width}=100$) with respect to the true ordered eigenvector centrality.
  }\label{fig:exp-er}
\vspace{-2mm}
\end{figure*}

\section{Finite rate analysis}\label{sec:ecp}

We present our main result on the performance of Algorithm~\ref{alg:bci}, where we state sufficient sampling conditions to preserve (with high probability) the relative centrality ordering of nodes.
In stating this result, we use the notation $\| \cdot \|_p$ to denote the $\ell_p$ norm when the argument is a vector and the related induced norm when the argument is a matrix.
Also, we denominate the eigenvalues and eigenvectors of $\bbC_y$ by $\beta_1\geq\ldots\geq\beta_n\geq 0$ and $\bbz_1,\bbz_2,\ldots,\bbz_n$, respectively. From our discussion following \eqref{eq:covariance} it holds that $\bbz_1 = \bbu$.

\begin{Theorem}\label{thm:main}
  Define $\mu:=n\| \bbu \|_\infty^2$, $\bbE = \bbC_y-\widehat{\bbC}_y^N$, and $\kappa=\|\bbC_y - \beta_1 {\bbu} {\bbu}^\top \|_\infty$.
  Assume that $\| \bby_i \|^2_2 \leq m$ for some constant $m$ and $\beta_1-\kappa=\Omega\left(\mu^2\|\bbE\|_\infty\right)$. 
  Consider any pair of nodes $i,j$ where $| u_i - u_j |>\alpha$. If
  \begin{equation}\label{eq:main}
  N\geq C\mu^4\left(\frac{t}{\alpha}\right)^2\frac{m\log n}{\beta_1(1-\kappa/\beta_1)^2}, \\
  \end{equation}
  for some absolute constant $C$ and $t>0$, then $\bbu$ and $\widehat{\bbu}$ preserve the relative order of $i$ and $j$ with probability at least $1-n^{-t^2}$.
\end{Theorem}
\begin{proof} \emph{(sketch)}
  First, we bound the infinity-norm of the matrix $\bbE$.
  By~\cite[Corollary 5.5.2]{vershynin2010introduction} and the equivalence of norms, under our assumptions, the following holds with probability at least $1-n^{-t^2}$:
  \begin{equation}
    \label{eq:normbound}
    \begin{aligned}
    &\text{If }N\geq C_0\mu^4\left(\frac{t}{\alpha'}\right)^2\frac{4\beta_1}{(\beta_1-\kappa)^2}m\log n \\
    &\text{then }\|E\|_\infty\leq\sqrt{n}\|E\|_2\leq\epsilon\sqrt{n}\beta_1
  \end{aligned}
  \end{equation}
  where $\epsilon=\frac{\alpha'(\beta_1-\kappa)}{2\beta_1\mu^2}$,
  and $C_0$ is an absolute constant.
  It follows from~\cite[Theorem 2.1]{fan2018perturbation}, up to the sign of $\widehat{\bbu}$,
  \begin{equation}
    \label{eq:perturbboundconst}
    \|\bbu-\widehat{\bbu}\|_\infty\leq\frac{C_1\mu^2\epsilon\beta_1}{\beta_1-\kappa},
  \end{equation}
  for some constant $C_1$.
  Furthermore, by applying the substitutions $\epsilon=\frac{\alpha'(\beta_1-\kappa)}{2\beta_1\mu^2}$, $C=4C_0C_1^2$, and $\alpha=C_1\alpha'$ in Equations~\ref{eq:normbound}~and~\ref{eq:perturbboundconst}, we get
  \begin{equation}
    \label{eq:mainproven}
    \begin{aligned}
    &\text{If }N\geq C\mu^4\left(\frac{t}{\alpha}\right)^2\frac{m\log n}{\beta_1(1-\kappa/\beta_1)^2} \\
    &\text{then }\|\bbu-\widehat{\bbu}\|_\infty\leq\frac{\alpha}{2}.
    \end{aligned}
  \end{equation}
  This completes the proof, as an element-wise perturbation of magnitude at most $\alpha/2$ is guaranteed to preserve the ordering of nodes whose centralities differ by more than $\alpha$.
\end{proof}

Theorem~\ref{thm:main} characterizes the sampling requirements of Algorithm~\ref{alg:bci} in terms of a desired \emph{resolution} $\alpha$, i.e., we determine the samples required to correctly order two nodes that differ by at least $\alpha$ in their centrality values.
As $\alpha$ decreases, the sampling requirement increases by a factor of $1/\alpha^2$, reflecting the difficulty of differentiating nodes whose centralities are very close together.
The assumption of ${\bm y}$ being bounded in a Euclidean ball of radius $\sqrt{m}$ holds for bounded inputs to a graph filter.
Considering the example of opinions in a social network, it is reasonable to assume that the measured graph signal is bounded, i.e. there is a limit to the extremes in opinion dynamics.

The result in Theorem~\ref{thm:main} differs from existing work in blind network inference \cite{schaub2018blind,wai2018blind}, where the performance of similar approaches for community detection are justified in terms of the alignment of principal subspaces between the true and sample covariance matrices, which is characterized by the Davis-Kahan $\sin\Theta$ Theorem \cite{daviskahan1970}.
For the ranking problem, we are concerned with element-wise perturbations of the leading eigenvectors, requiring the use of more modern statistical results \cite{fan2018perturbation}.

\subsection{Graph spectra and eigenvector delocalization}\label{sec:randomgraphs}

The spectrum of the covariance is determined by the spectrum of (the adjacency matrix of) the graph itself and the frequency response of the filter.
To illustrate this, we consider an \ER graph of $n$ nodes, where each edge exists independently with probability $p$.
In the dense regime, where $pn\geq C\log n, C\geq 1$, it is well-known that as the graph becomes sufficiently large, the leading eigenvalue converges to $pn$, and $\lambda_2({\bbA})={\cal O}(\sqrt{pn})$ \cite{farkas2001spectra}.
For simplicity, consider the case where the graph filter is equal to the adjacency matrix: ${\cal H}(\bbA)={\bbA}$.
So, the covariance matrix is the square of the adjacency matrix, and thus has leading eigenvalue $\beta_1\approx p^2n^2$.

In order to bound $\kappa$ in~\eqref{eq:main} for this specific graph type we apply the equivalence of norms to yield
\begin{equation}
  \label{eq:er-kappa}
 \kappa=\|\bbC_y-\beta_1\bbu\bbu^\top\|_\infty\leq\sqrt{n} \beta_2 ={\cal O}(n^{3/2}p).
\end{equation}
We emphasize that the graph filter has significant influence on $\kappa$; a filter that strongly decreases the ratio $\beta_2/\beta_1$ will improve the performance of Algorithm~\ref{alg:bci}.
For instance, an `ideal filter' that annihilates the lower spectrum of the adjacency matrix will yield optimal performance, since the ratio $\kappa/\beta_1$ would be zero in~\eqref{eq:main}.

We proceed to analyze the value $\mu=n\|\bbu\|_\infty^2$.
This measures concentrated the eigenvector centrality is in the most central node, i.e., $\mu$ reflects the localization of the leading eigenvector.
We consider the result of~\cite{he2019delocalization} which shows that in the dense regime, $\|\bbu\|_\infty^2\leq n^{-1+g(n)}$, for some $g(n)=o(1)$.
Thus, $\mu\leq n^{g(n)}$, so the eigenvector centrality becomes increasingly delocalized as $n\to\infty$.
Moreover, the largest possible gap between any nodes $i,j$ is attained when a node has eigenvector centrality $0$.
Thus, the largest meaningful $\alpha$ in the context of Theorem~\ref{thm:main} is less than $n^{g(n)/2-1/2}$, since $\|\bbu\|_\infty\leq \sqrt{n^{-1+g(n)}}=\sqrt{n^{g(n)}}/\sqrt{n}=n^{g(n)/2-1/2}$.
So, to preserve the order of the \textit{most central} and \textit{least central} nodes in this scenario, it can be shown that
\begin{equation}
  \label{eq:er-localized-sampling}
  N\geq Ct^2m\frac{n^{o(1)+1}}{\log n}
\end{equation}
samples are sufficient, under the same conditions as Theorem~\ref{thm:main}.
This quantity tends to infinity with $n$, reflecting the increasing difficulty of the ranking problem as the eigenvector centrality becomes increasingly delocalized.
We further assess empirically the tightness of bound \eqref{eq:er-localized-sampling} in our numerical experiments.

\begin{figure}[t!]
  \centering
  \resizebox{0.6\linewidth}{!}{\includegraphics[width=\linewidth]{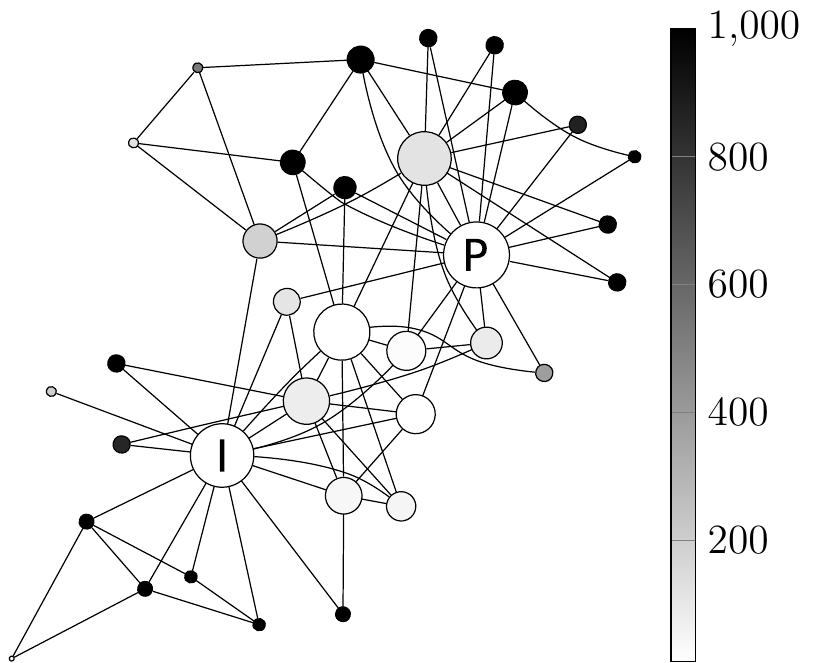}}
\vspace{-2mm}
  \caption{
    Karate club network, with the president (P) and instructor (I) labelled.
    Nodes are scaled proportionally to their eigenvector centrality, and colored according to the number of samples sufficient for proper ranking.
    }\label{fig:exp-zkc}
\vspace{-2mm}
\end{figure}

\section{Numerical experiments}\label{sec:exp}

We demonstrate the performance of our algorithm on \ER and \BA graph models, as well as the karate club network~\cite{albert2002statistical,karate,Newman2010}.
The graph filter used is of the form
\begin{equation}\label{eq:exp-filter}
  {\cal H}\left({\bbA}\right)=\sum_{k=0}^4\left(\frac{{\bbA}}{\lambda_1({\bbA})}\right)^k.
\end{equation}
We evaluate our results with the absolute value of the Spearman correlation, a measure of monotonicity between two ranking functions~\cite{cordernonparam}.

\textbf{\ER graphs.}
We consider an \ER random graph model with graph size $n$ and edge probability $p$.
We first show how the edge probability influences the performance of Algorithm~\ref{alg:bci} in the ranking task.
For fixed $n=500$, we vary $p$ from $0.1$ to $0.9$, and compute the Spearman correlation of the inferred eigenvector centrality against the true eigenvector centrality for an increasing number of samples $N$.

As pictured in Fig.~\letterref{fig:exp-er}{A}, the performance of the algorithm varies inversely with the parameter $p$.
This is expected, as dense graphs have more delocalized eigenvectors, and thus have small average distances between centralities.
Then, by Theorem~\ref{thm:main}, the ranking problem becomes increasingly difficult due to small values of $\alpha$.

Looking past the edge probability, we consider the relationship between the number of samples $N$ and the graph size $n$ in the dense \ER graph regime.
For each $n$, we set $p=4\log n/n$.
We consider the algorithm performance against the number of samples taken, normalized by $\log n/n$, by the derivation in \eqref{eq:er-localized-sampling}.
Fig.~\letterref{fig:exp-er}{B} shows that for each graph size $n$, the performance closely matches over the normalized sample size $N\log n/n$, as expected.
Barring a slight deviation for small $n$, due to the unknown $n^{o(1)}$ term in \eqref{eq:er-localized-sampling} unaccounted for by the normalized sample size, this demonstrates the tightness of the sampling requirements for dense \ER graphs and, for this setting, confirms the practical validity of Theorem~\ref{thm:main}.

\textbf{\BA graphs.}
To illustrate our algorithm on a more realistic type of graph, we consider a \BA graph model of size $n=500$ with parameter $m=m_0=3$~\cite{albert2002statistical}.
As shown in Fig.~\letterref{fig:exp-er}{C}, while \ER graphs have centralities that are all very close to $1/\sqrt{n}$, \BA graphs are characterized by a highly localized centrality structure.

Due to this unevenness in the distribution of the eigenvector centrality, we evaluate the Spearman correlation over a sliding window, i.e., for groups of nodes with consecutive ground-truth centrality rankings.
This highlights the varying sensitivity to perturbation based on differences in centrality, as described in Theorem~\ref{thm:main}.

As shown in Fig.~\letterref{fig:exp-er}{D}, performance over every region of the graph (where ``region'' refers to nodes of similar centrality) improves as the number of samples increases, with the most central nodes being properly ranked with fewer samples than the rest.
In particular, for the top 100 nodes (nodes 400-500), we achieve a Spearman correlation of 0.84 with as few as $\approx 600$ samples, while the next lower block (nodes 300-400) needs 5000 samples to achieve a Spearman correlation of 0.81.
This illustrates the highly centralized nature of the \BA model, reflected in the power-law distribution of node degrees~\cite{albert2002statistical}.
The most central nodes in the graph are more separated from other nodes on average, compared to the nodes with lower eigenvector centrality (see Figure~\letterref{fig:exp-er}{C}), making the ranking task easier for these regions of the graph.
In terms of Theorem~\ref{thm:main}, the sampling requirements here are low due to a large tolerance $\alpha$.

\textbf{Zachary's Karate Club.}
Finally, we show the application of Algorithm~\ref{alg:bci} to the karate club network~\cite{karate,girvan2002community}.
To evaluate its performance, for an inferred ranking $\widehat{r}$, we consider node $i$ to be correctly ranked with respect to the true centrality ranking $r$ if $|\widehat{r}_i-r_i|\leq 1$.
Moreover, we consider a number of samples $N$ to be sufficient to rank node $i$ if for all sample counts $N'$ such that $1000\geq N'\geq N$, node $i$ is ranked correctly with probability greater than $0.95$.
If this is not attained, we saturate the required number of samples at 1000.

The results of this experiment are shown in Fig.~\ref{fig:exp-zkc}.
It is clear that the nodes with higher eigenvector centrality require fewer samples to be properly ranked, as they are more separated from other nodes.
More precisely, the most central nodes are the Instructor (I) with eigenvector centrality 0.36, and the President (P) with eigenvector centrality 0.37, requiring 80 and 10 samples respectively for proper ranking.
The least central nodes, found on the periphery of the graph, have very similar centralities, and are difficult to distinguish.
For instance, 5 nodes in this network have centrality close to 0.101, making the sampling requirements for distinguishing them from other similarly-ranked nodes very high due to small tolerance~$\alpha$.

\section{Conclusions and future work}\label{sec:disc}

We have considered the problem of ranking the nodes of an unobserved graph based on their eigenvector centrality, given only a set of graph signals regularized by the graph structure.
Using tools from matrix perturbation theory, we characterize the $\ell_\infty$ perturbation of the leading eigenvector of the empirical covariance matrix in terms of the graph filter's spectrum and number of samples taken.
We then present an analysis of dense \ER graphs in this regime, showing the interplay between the ranking problem and eigenvector delocalization.
These theoretical results are then demonstrated on \ER and \BA random graphs, highlighting the influence of the graph structure on the performance of our approach.
Finally, we demonstrate the sampling requirements for approximate ranking on the karate club network.

Future research avenues include the consideration of additional (spectrum-based) centrality measures, such as Katz and PageRank centralities, as well as the relaxation of some of the assumptions in our analysis, such as the whiteness assumption of the inputs in the generation of our observations [cf.~\eqref{eq:graphfilter_2}].

\newpage

\small
\bibliographystyle{IEEEtran}
\bibliography{ref}

\begin{thebibliography}{10}
\providecommand{\url}[1]{#1}
\csname url@samestyle\endcsname
\providecommand{\newblock}{\relax}
\providecommand{\bibinfo}[2]{#2}
\providecommand{\BIBentrySTDinterwordspacing}{\spaceskip=0pt\relax}
\providecommand{\BIBentryALTinterwordstretchfactor}{4}
\providecommand{\BIBentryALTinterwordspacing}{\spaceskip=\fontdimen2\font plus
\BIBentryALTinterwordstretchfactor\fontdimen3\font minus
  \fontdimen4\font\relax}
\providecommand{\BIBforeignlanguage}[2]{{%
\expandafter\ifx\csname l@#1\endcsname\relax
\typeout{** WARNING: IEEEtran.bst: No hyphenation pattern has been}%
\typeout{** loaded for the language `#1'. Using the pattern for}%
\typeout{** the default language instead.}%
\else
\language=\csname l@#1\endcsname
\fi
#2}}
\providecommand{\BIBdecl}{\relax}
\BIBdecl

\bibitem{Strogatz2001}
S.~H. Strogatz, ``Exploring complex networks,'' \emph{Nature}, vol. 410, no.
  6825, pp. 268--276, Mar. 2001.

\bibitem{Newman2010}
M.~E.~J. Newman, \emph{{N}etworks: {A}n {I}ntroduction}.\hskip 1em plus 0.5em
  minus 0.4em\relax Oxford University Press, USA, Mar. 2010.

\bibitem{Jackson2010}
M.~O. Jackson, \emph{Social and Economic Networks}.\hskip 1em plus 0.5em minus
  0.4em\relax Princeton university press, 2010.

\bibitem{garroway_2008_applications}
C.~J. Garroway, J.~Bowman, D.~Carr, and P.~J. Wilson, ``Applications of graph
  theory to landscape genetics,'' \emph{Evolutionary Appl.}, vol.~1, no.~4, pp.
  620--630, 2008.

\bibitem{holme_2002_attack}
P.~Holme, B.~J. Kim, C.~N. Yoon, and S.~K. Han, ``Attack vulnerability of
  complex networks,'' \emph{Phys. Rev. E}, vol.~65, p. 056109, May 2002.

\bibitem{beauchamp_1965_improved}
M.~A. Beauchamp, ``An improved index of centrality,'' \emph{Behavioral Sc.},
  vol.~10, no.~2, pp. 161--163, 1965.

\bibitem{freeman_1977_set}
L.~C. Freeman, ``A set of measures of centrality based on betweenness,''
  \emph{Sociometry}, vol.~40, no.~1, pp. 35--41, 1977.

\bibitem{segarra_2016_stability}
S.~{Segarra} and A.~{Ribeiro}, ``Stability and continuity of centrality
  measures in weighted graphs,'' \emph{IEEE Trans. Signal Process.}, vol.~64,
  no.~3, pp. 543--555, Feb 2016.

\bibitem{bonacich_1972_factoring}
P.~Bonacich, ``Factoring and weighting approaches to status scores and clique
  identification,'' \emph{J. Math. Soc.}, vol.~2, no.~1, pp. 113--120, 1972.

\bibitem{friedman_2008_sparse}
J.~Friedman, T.~Hastie, and R.~Tibshirani, ``Sparse inverse covariance
  estimation with the graphical lasso,'' \emph{Biostatistics}, vol.~9, no.~3,
  pp. 432--441, 2008.

\bibitem{lake_2010_discovering}
B.~M. Lake and J.~B. Tenenbaum, ``Discovering structure by learning sparse
  graphs,'' in \emph{Annual Cognitive Sc. Conf.}, Aug. 2010, pp. 778--783.

\bibitem{meinshausen_2006_high}
N.~Meinshausen and P.~Buhlmann, ``High-dimensional graphs and variable
  selection with the lasso,'' \emph{Ann. of Stat.}, vol.~34, no.~3, pp.
  1436--1462, 2006.

\bibitem{egilmez_2017_graph}
H.~E. {Egilmez}, E.~{Pavez}, and A.~{Ortega}, ``Graph learning from data under
  {Laplacian} and structural constraints,'' \emph{IEEE J. Sel. Topics Signal
  Process.}, vol.~11, no.~6, pp. 825--841, Sep. 2017.

\bibitem{shen_2017_kernel}
Y.~{Shen}, B.~{Baingana}, and G.~B. {Giannakis}, ``Kernel-based structural
  equation models for topology identification of directed networks,''
  \emph{IEEE Trans. Signal Process.}, vol.~65, no.~10, pp. 2503--2516, May
  2017.

\bibitem{dong2019tutorial}
X.~{Dong}, D.~{Thanou}, M.~{Rabbat}, and P.~{Frossard}, ``Learning graphs from
  data: A signal representation perspective,'' \emph{IEEE Signal Process.
  Mag.}, vol.~36, no.~3, pp. 44--63, May 2019.

\bibitem{kalofolias_smooth_2016}
V.~Kalofolias, ``How to learn a graph from smooth signals,'' in \emph{Intl.
  Conf. Artif. Intel. Stat. (AISTATS)}, Jan. 2016, pp. 920--929.

\bibitem{segarra_topo_2017}
S.~Segarra, A.~G. Marques, G.~Mateos, and A.~Ribeiro, ``Network topology
  inference from spectral templates,'' \emph{IEEE Trans. Signal Inf. Process.
  Netw.}, vol.~3, no.~3, pp. 467--483, Aug. 2017.

\bibitem{mateos_2019_connecting}
G.~{Mateos}, S.~{Segarra}, A.~G. {Marques}, and A.~{Ribeiro}, ``Connecting the
  dots: Identifying network structure via graph signal processing,'' \emph{IEEE
  Signal Process. Mag.}, vol.~36, no.~3, pp. 16--43, May 2019.

\bibitem{zhu2019consensus}
Y.~Zhu, M.~Schaub, A.~Jadbabaie, and S.~Segarra, ``Network inference from
  consensus dynamics with unknown parameters,'' \emph{IEEE Trans. Signal Inf.
  Process. Netw. (under review)}, 2019.

\bibitem{schaub2018blind}
M.~T. Schaub, S.~Segarra, and J.~Tsitsiklis, ``Blind identification of
  stochastic block models from dynamical observations,'' \emph{SIAM J. Math.
  Data Science (SIMODS) (under review)}, 2019.

\bibitem{wai2018blind}
H.~{Wai}, S.~{Segarra}, A.~E. {Ozdaglar}, A.~{Scaglione}, and A.~{Jadbabaie},
  ``Community detection from low-rank excitations of a graph filter,'' in
  \emph{IEEE Intl. Conf. Acoust., Speech and Signal Process. (ICASSP)}, Apr.
  2018, pp. 4044--4048.

\bibitem{wai2019hidden}
H.~{Wai}, Y.~C. {Eldar}, A.~E. {Ozdaglar}, and A.~{Scaglione}, ``Community
  inference from graph signals with hidden nodes,'' in \emph{IEEE Intl. Conf.
  Acoust., Speech and Signal Process. (ICASSP)}, May 2019, pp. 4948--4952.

\bibitem{schaub2019partition}
M.~T. {Schaub}, S.~{Segarra}, and H.~{Wai}, ``Spectral partitioning of
  time-varying networks with unobserved edges,'' in \emph{IEEE Intl. Conf.
  Acoust., Speech and Signal Process. (ICASSP)}, May 2019, pp. 4938--4942.

\bibitem{Hoffmann2018}
T.~Hoffmann, L.~Peel, R.~Lambiotte, and N.~S. Jones, ``Community detection in
  networks with unobserved edges,'' \emph{arXiv preprint arXiv:1808.06079},
  2018.

\bibitem{ruggeri2019centrality}
N.~Ruggeri and C.~{de Bacco}, ``Sampling on networks: {Estimating} eigenvector
  centrality on incomplete graphs,'' \emph{arXiv preprint arXiv:1908.00388},
  Aug. 2019.

\bibitem{shao2017centrality}
H.~Shao, M.~Mesbahi, D.~Li, and Y.~Xi, ``Inferring centrality from network
  snapshots,'' \emph{Scientific Reports}, vol.~7, Jan. 2017.

\bibitem{Olfati-Saber2007}
R.~Olfati-Saber, J.~A. Fax, and R.~M. Murray, ``Consensus and cooperation in
  networked multi-agent systems,'' \emph{Proc. IEEE}, vol.~95, no.~1, pp.
  215--233, Mar. 2007.

\bibitem{Masuda2017}
N.~Masuda, M.~A. Porter, and R.~Lambiotte, ``Random walks and diffusion on
  networks,'' \emph{Physics Reports}, vol. 716, pp. 1--58, Nov. 2017.

\bibitem{vershynin2010introduction}
R.~Vershynin, ``Introduction to the non-asymptotic analysis of random
  matrices,'' in \emph{Compressed Sensing: Theory and Applications}, Y.~C.
  Eldar and G.~Kutyniok, Eds.\hskip 1em plus 0.5em minus 0.4em\relax Cambridge
  university press, 2012.

\bibitem{fan2018perturbation}
J.~Fan, W.~Wang, and Y.~Zhong, ``An $\ell_\infty$ perturbation bound and its
  application to robust covariance estimation,'' \emph{J. Mach. Learn. Res},
  vol.~18, no. 207, pp. 1--42, 2018.

\bibitem{daviskahan1970}
C.~Davis and W.~Kahan, ``The rotation of eigenvectors by a perturbation. iii,''
  \emph{SIAM Journal on Numerical Analysis}, vol.~7, no.~1, pp. 1--46, 1970.

\bibitem{farkas2001spectra}
I.~J. Farkas, I.~Der\'enyi, A.-L. Barab\'asi, and T.~Vicsek, ``Spectra of
  ``real-world'' graphs: Beyond the semicircle law,'' \emph{Phys. Rev. E},
  vol.~64, p. 026704, Jul. 2001.

\bibitem{he2019delocalization}
Y.~He, A.~Knowles, and M.~Marcozzi, ``Local law and complete eigenvector
  delocalization for supercritical {Erd\H{o}s}-{R\'{e}nyi} graphs,'' \emph{Ann.
  of Stat. (to appear)}, 2019.

\bibitem{albert2002statistical}
R.~Albert and A.-L. Barab\'{a}si, ``Statistical mechanics of complex
  networks,'' \emph{Reviews of Modern Physics}, vol.~74, pp. 47--97, Jan. 2002.

\bibitem{karate}
W.~Zachary, ``An information flow model for conflict and fission in small
  groups,'' \emph{Journal of Anthropological Research}, vol.~33, no.~4, pp.
  452--473, 1977.

\bibitem{cordernonparam}
G.~W. Corder and D.~I. Foreman, \emph{Comparing Variables of Ordinal or
  Dichotomous Scales: Spearman Rank-Order, Point-Biserial, and Biserial
  Correlations}.\hskip 1em plus 0.5em minus 0.4em\relax John Wiley \& Sons,
  Ltd, 2011, ch.~7, pp. 122--154.

\bibitem{girvan2002community}
M.~Girvan and M.~E.~J. Newman, ``Community structure in social and biological
  networks,'' \emph{Proc. of the National Academy of Sciences}, vol.~99,
  no.~12, pp. 7821--7826, 2002.

\end{thebibliography}

\end{document}